\documentclass[a4paper,12pt]{article}
\usepackage[a4paper, total={6.5in, 8.5in}]{geometry}
\usepackage[utf8]{inputenc}
\usepackage[english]{babel}
\usepackage{a4wide}
\usepackage{booktabs}
\usepackage{amssymb}
\usepackage{float}
\usepackage{graphicx}
\usepackage{tikz}
\usepackage{xcolor}
\usepackage{amsthm}
\usepackage{amsmath}
\usepackage{mathrsfs}
\usetikzlibrary{positioning}
\usetikzlibrary{graphs,graphs.standard}
\usepackage{comment}

\usepackage[textwidth=30mm]{todonotes}
\newcommand{\AP}[1]{\todo[color=white, author=\textbf{Arkadi},inline]{\small #1\\}}

\newcommand{\pr}{\mathbb{P}}

\usepackage{subcaption}

\theoremstyle{theorem}
\newtheorem{thm}{Theorem}[section]
\newtheorem{cor}[thm]{Corollary}
\newtheorem{lem}[thm]{Lemma}

\newtheorem{exmp}[thm]{Example}
\newtheorem{defn}[thm]{Definition}

\numberwithin{equation}{section}

\newcommand{\N}{\mathbb{N}}

\title{The Expected Number of Pairwise Stable Networks\thanks{We would like to thank Philippe Bich, Marcus Pivato, Marco Scarsini and Xavier Venel for helpful discussions. Excellent research assistance by Deniz Sava\c{s} is gratefully acknowledged.}}

\author{P. Jean--Jacques Herings\thanks{Tilburg University, E-Mail: \texttt{P.J.J.Herings@tilburguniversity.edu}} \and Arkadi Predtetchinski\thanks{Maastricht University, E-Mail: \texttt{a.predtetchinski@maastrichtuniversity.nl}} \and Christian Seel \thanks{Maastricht University, E-Mail: \texttt{c.seel@maastrichtuniversity.nl}}}

\begin{document}

\maketitle

\begin{abstract}
\noindent This paper studies probabilistic properties of pairwise stability for a network model where individual utilities are random variables. We study the probability that a given network is pairwise stable and the expected number of pairwise stable networks. We provide a closed-form solution for the latter number.

As the evaluation of the exact expression is computationally challenging for large populations, we provide tractable lower and upper bounds for this expression which allow us to pin down the asymptotic behavior of the expected number of pairwise stable networks up to a multiplicative constant. This asymptotic behavior is described by the number of networks $ 2^{n(n-1)/2} $ times $ (2/n+1)^{n} $, a sequence that tends to infinity fast. We normalize the number of pairwise stable networks by this sequence and show that the variance of the normalized number of pairwise stable networks converges to zero as $ n $ tends to infinity.

We conclude that almost surely the number of pairwise stable networks tends to infinity, while the fraction of pairwise stable networks tends to $ 0 $ as $ n $ goes to infinity.
\end{abstract}
\newpage

\section{Introduction}
Pairwise stability as introduced in Jackson and Wolinsky (1996) is the most influential solution concept for economic network formation. A network is pairwise stable if no individual benefits from deleting a link and no pair of individuals benefits by creating a link between them. On the one hand, pairwise stability is often seen a weak/permissive solution concept (Jackson and van den Nouweland, 2005). On the other hand, it is easy to provide examples where no pairwise stable network exists and thus the literature developed sufficient conditions for existence; see e.g., Jackson and Watts (2001) and Hellmann (2013). In this work, we study pairwise stability from a probabilistic point of view. Our approach allows us to shed new light on the issue of (non-)existence and the expected number of pairwise stable networks.

As in Jackson and Wolinsky (1996), a network consists of bilateral links in a group of $n$ individuals. Each individual's utility for any given network is drawn i.i.d. from an atomless distribution. We derive two groups of results. The first group of results focuses on the probability that a given network is pairwise stable and on the expected number of pairwise stable networks. The second group of results studies the probability that two given networks are both pairwise stable, as well as the variance of the number of pairwise stable networks.

We start out by looking at the probability for a given network to be pairwise stable. In particular, we show that the probability that a network is pairwise stable decreases when a link is added. In other words, the probability for a network to be pairwise stable is monotone with respect to set inclusion, with the empty network having the highest probability to be pairwise stable and the complete network the lowest.

In the next step, we derive an exact expression for the expected number of pairwise stable networks which relies on the so-called seniority degrees. Relative to an individual $i$, the individuals with a lower index are called $i$'s juniors, and those with a higher index are called $i$'s seniors. The seniority degree of individual $i$ in a network is calculated by awarding $i$ one point for every link with $i$'s seniors, one point for every missing link with $i$'s juniors, and one point for free. The seniority degree of the network is then obtained as the product of all the individual seniority degrees. The expected number of pairwise stable networks is now equal to the sum over all graphs of the inverse of the seniority degrees of the graphs.

Evaluating the exact expression for the expected number of pairwise stable networks becomes computationally challenging already for populations of modest size. When $n=7$, the seniority degree has to be calculated for a total of $2^{28}$ networks. Thus, we proceed to derive tractable lower and upper bounds for the exact expression.

To obtain a lower bound on the expected number of pairwise stable networks, we apply a logarithmic transformation to the product of the individual seniority degrees, use Jensen's inequality, and thus reduce the problem to calculating the expectation of the logarithm of a Bernoulli random variable.

To derive the upper bound we obtain another expression for the expected number of pairwise stable networks similar to that found in Knuth (1976) for the expected number of stable matchings. It involves an iterated integral, where the integrand is dominated by a particular product of independent random variables, leading to an upper bound.

These bounds reveal that the expected number of pairwise stable networks is of the order of
\[
c_{n} = 2^{\tfrac{n(n-1)}{2}}(\tfrac{2}{n+1})^{n},
\]
where the first term in the definition of $c_{n}$ is equal to the cardinality of the set of all networks with $n$ agents. The sequence $ c_{n} $ tends to infinity fast when $ n $ does.

If we divide the number of pairwise stable networks by $ c_{n}, $ we obtain the normalized number of pairwise stable networks. We show that the normalized expected number of pairwise stable networks is asymptotically bounded below by $\sqrt{e}$ and above by $e$.

We now turn to the second group of results. We relate the probability for a pair of networks to be simultaneously pairwise stable to the Hamming distance between these networks. If the Hamming distance between two networks is 1, they cannot be pairwise stable at the same time with positive probability. If the Hamming distance between two networks is 2, then the two events that they are pairwise stable are positively dependent, while if the Hamming distance is 3 or greater, then the events that the respective networks are pairwise stable are independent.

As $n$ gets large, the fraction of pairs of networks at a Hamming distance of at least 3 goes to one. Thus, as $n$ gets large, the variance of the normalized number of pairwise stable networks goes to zero. As an implication, the number of pairwise stable networks almost surely tends to infinity as $ n $ tends to infinity. In particular, for any given number, the probability that there exist at least that number of pairwise stable networks goes to 1 as $ n $ goes to infinity. These results therefore give a probabilistic foundation to the view of pairwise stability as a weak concept.

The probabilistic approach used in this paper has previously been adopted in several other contexts, most notably in strategic games. 
In an early contribution, Goldman (1957) studies random zero-sum games and explicitly computes the probability that such a game admits a pure strategy Nash equilibrium. Goldberg, Goldman, and Newman (1968) and Dresher (1970) extend this analysis to $n$-player games, calculate the probability for a game to have a pure strategy Nash equilibrium, and compute the limit of this probability as the number of the strategies becomes large. 
Rinott and Scarsini (2000) allow correlation in payoffs, and also consider the asymptotic distribution of the number of pure strategy Nash equilibria as the number of players becomes large. Amiet, Collevecchio, Scarsini, and Zhong (2021) study the random geometry of pure strategy Nash equilibria in games with a large number of players where each player has two strategies.

The literature on random games is not restricted to pure strategy Nash equilibrium as a solution concept. For instance, McLennan (2005) and McLennan and Berg (2005) focus on the expected number of mixed strategy Nash equilibria, B\a'{a}r\a'{a}ny, Vempala, and Vetta (2007) study convergence of algorithms to compute Nash equilibria in random games, while Amiet, Collevecchio, Scarsini, and Zhong (2021) study the probability of convergence of best-response dynamics to a Nash equilibrium.

Matching theory is another area where the probabilistic approach has been fruitfully applied. Pittel (1989) studies a marriage market with an equal number of men and women and assumes that the preferences over potential partners are randomly generated. He pins down the asymptotic behavior of the expected number of stable matchings as the number of individuals becomes large. Ashlagi, Kanoria, and Leshno (2017) extend the analysis to marriage markets where the numbers of men and women differ and show that such an asymmetry leads to an essentially unique stable matching.

Differences in context lead to different results and require the use of different techniques. In a random strategic game, all pure strategy profiles are ex-ante identical, and thus equally likely to be a pure strategy Nash equilibrium. In contrast, pairwise stability is an asymmetric concept where link deletions are unilateral, while link creations are bilateral. This leads to different statistical properties across networks and a decreasing probability that a network is pairwise stable as more links are added.

In comparison to random matching models, the probability space on the set of networks is much larger as each individual can have multiple links. Moreover, the solution concepts differ. Pairwise stability of a network considers adding or deleting one link at a time. In a matching, adding a new link might affect the current partners and hence must be accompanied by the deletion of up to two links. Due to Gale and Shapley (1962), a stable matching is guaranteed to exist for any preference profile. In contrast, it is easy to give examples without any pairwise stable network.

The paper is organized as follows. Section~2 introduces the network model with random utilities. In Section~3, we study the probability that a given network is pairwise stable.  We then provide our main results on the expected number of pairwise stable networks in Section~4 and the variance of the normalized number of pairwise stable networks in Section~5. We provide a short discussion in Section~6. 
\section{The Network Model with Random Utilities}
In this section, we introduce our probabilistic model of networks.

Let $\N = \{1,2,\ldots\}$. For $n \in \N$, let $I_{n} = \{1,\ldots,n\}$ be a population of size $n$ and let $L_{n} = \{\{i,j\}: 1 \leq i < j \leq n\}$ be the set of all possible links in the population $I_{n}$. A typical link is written as $ij.$ The cardinality of $L_{n}$ is given by $\ell_{n} = n(n-1)/2$. A network in the population $I_{n}$ is any subset of $L_{n}$. We let $G_{n} = 2^{L_{n}}$ denote the set of networks in the population $I_{n}. $
A \textit{network model of size $n$} is a function $ u : I_{n} \times G_{n} \to [0,1]$. Let $ U_{n} = [0,1]^{I_{n} \times G_{n}} $ denote the set of network models of size $n$.

\begin{defn}[Pairwise Stability, Jackson and Wolinsky (1996)]\label{defn.ps}
{\rm Let $ n \in \N $ and $ u \in U_{n}. $ A network $ g \in G_{n} $ is \emph{pairwise stable} if (i) for every $ij \in g$, $ u_i(g) \geq u_i(g-ij)$ and $u_j(g) \geq u_j(g-ij)$ and (ii) for every $ij \in L_n \setminus g$, there exists $k \in \{i,j\}$ such that $u_k(g) \geq u_k(g+ij)$.}
\end{defn}

In the above notion of pairwise stability, condition (i) requires immunity to unilateral link deletion and condition (ii) requires immunity to bilateral link addition. The requirements for link deletion and link addition are therefore not symmetric.

In this paper, we endow $U_{n}$ with the product $ \mathbb{P}_{n} $ of uniform distributions on $ [0,1] $ and refer to $ (U_{n}, \mathbb{P}_{n}) $ as a {\it network model with random utilities.} Put equivalently, we assume that all coordinate functions $u_{i}(g)$, with $i \in I_{n}$ and $g \in G_{n}$, are independently and uniformly distributed on the interval $[0,1]$.

We use uniform distributions on $ [0,1] $ only for notational convenience. To see that our results extend to arbitrary atomless distributions, take such a distribution $ F $ and transform the utility of every agent $ i \in I_{n} $ by $ F. $ Since $ F $ is a strictly positive monotonic transformation, up to a set of points with measure zero, the set of pairwise stable networks is unaffected by the transformation. 

In accordance with the standard terminology in probability theory, subsets of $U_{n}$ are henceforth called events. Note that the event that an individual obtains the same utility at two different networks has probability zero. This implies that, almost surely, for all $ i \in I_{n}, $ the ranking of the networks is strict. For this paper, it will therefore not matter if we define pairwise stability as in Definition~2.1 or use the more demanding version where link addition takes place as soon as one player benefits from doing so and the other player is at least indifferent.

We now define events that are particularly important in our study, including the event that a given network is dominated by another one in the sense of conditions~(i) and (ii) in Definition~\ref{defn.ps} and the event that a given network is pairwise stable.

Let $n \in \N$, $g_{0} \in G_n$, and $ij \in L_{n}$. We define $g_{1} = g_{0} + ij$ in case $ij$ is not in $g_{0}$ and $g_{1} = g_{0} - ij$ in case it is and denote the event that $ g_{1} $ dominates $ g_{0} $ by $ D(g_{0},g_{1}). $ The complement of this event is denoted by $ D^{\rm c}(g_{0},g_{1}). $
\[
\begin{array}{rcll}
D(g_{0},g_{0} - ij) & \hspace*{-2mm} = \hspace*{-2mm} & \{u \in U: u_{i}(g_{0}) < u_{i}(g_{0} - ij)\text{ or }u_{j}(g_{0}) < u_{j}(g_{0} - ij)\}, & ij \in g_{0}\\
D(g_{0},g_{0} + ij) & \hspace*{-2mm} = \hspace*{-2mm} & \{u \in U: u_{i}(g_{0}) < u_{i}(g_{0} + ij)\text{ and }u_{j}(g_{0}) < u_{j}(g_{0} + ij)\}, & ij \in L_n \setminus g_{0}.
\end{array}
\]
For $g \in G_{n},$ let us write $S_{g}$ for the event that $g$ is pairwise stable:
\[
S_{g} = \bigcap_{ij \in g} D^{\rm c}(g,g-ij) \bigcap_{ij \in L_{n} \setminus g} D^{\rm c}(g,g+ij).
\]
The next sections are devoted to the probability that the event $ S_{g} $ occurs and the expected value of the number of networks for which this event occurs.
\section{The Probability that a Network Is Pairwise Stable}\label{sec.structure}


In this section, we derive an integral formula for the probability that a given network is pairwise stable. We use it to show that networks with more links are less likely to be stable and to compute the probability that the complete network is pairwise stable.

We first derive some facts about the conditional probability that a network is pairwise stable.
\begin{lem}\label{lem:conditionalprobs}
Consider the network model with random utilities $ (U_{n}, \mathbb{P}_{n}) $ and let $ g \in G_{n}. $
\begin{enumerate}
\item Given $ u(g) \in [0,1]^{I_{n}}, $ it holds that
\[\begin{array}{rcll}
\mathbb{P}_{n}(D^{\rm c}(g,g-ij)|u(g)) & = & u_{i}(g) u_{j}(g), & ij \in g,\\
\mathbb{P}_{n}(D^{\rm c}(g,g+ij)|u(g)) & = & 1 - (1-u_{i}(g)) (1-u_{j}(g)), & ij \in L_{n} \setminus g.
\end{array}
\]
\item Given $ u(g) \in [0,1]^{I_{n}}, $ the events in the collection $\{D(g,g - ij): ij \in g\} \cup \{D(g,g + ij): ij \in L_{n} \setminus g\} $ are conditionally independent.
\item Given $ u(g) \in [0,1]^{I_{n}}, $ the conditional probability that $g$ is pairwise stable is
\begin{equation}\label{eqn.condstab}
\mathbb{P}_{n}(S_{g}|u(g)) = \prod_{ij \in g} u_{i}(g) u_{j}(g) \prod_{ij \in L_{n} \setminus g} [1 - (1-u_{i}(g)) (1-u_{j}(g))].
\end{equation}
For every $ i \in I_{n}, $ this probability is non-decreasing in $ u_{i}(g). $
\item Given $ u(N(g)) = (u(g-ij)_{ij \in g}, u(g+ij)_{ij \in L_{n} \setminus g}) \in [0,1]^{I_{n} \times L_{n}}, $ the conditional probability $ \mathbb{P}_{n}(S_{g}|u(N(g))) $ is non-increasing in $ u_{i}(g - ij), $ $ ij \in g, $ and non-increasing in $u_{j}(g + ij), $ $ ij \in L_{n} \setminus g. $
\end{enumerate}
\end{lem}
\begin{proof}
Direct computation using the uniform distribution on $ [0,1] $ yields (1). To see (2), note that, given $ u(g), $ for each $ ij \in g $, the event $ D(g,g - ij) $ is measurable with respect to the sigma-algebra $ \Sigma_{g,ij} $ generated by the pair $ u_{i}(g - ij), u_{j}(g - ij) $ of stochastic variables and the event $ D(g,g + ij) $ is measurable with respect to the sigma-algebra $ \Sigma_{g,ij} $ generated by the pair $ u_{i}(g + ij), u_{j}(g + ij) $ of stochastic variables. The collection $\{u_{i}(g - ij), u_{j}(g - ij) : ij \in g\} \cup \{u_{i}(g + ij), u_{j}(g + ij) : ij \in L_{n} \setminus g\}$ consists of exactly $2\ell_{n}$ independent stochastic variables. It follows that the sigma-algebras in the collection $\{\Sigma_{g,ij}:ij \in L_{n}\}$ are independent. A straightforward calculation shows that Lemma~\ref{lem:conditionalprobs}.3 follows from Lemma~\ref{lem:conditionalprobs}.1 and Lemma~\ref{lem:conditionalprobs}.2. A formula analogous to (\ref{eqn.condstab}) gives Lemma~\ref{lem:conditionalprobs}.4.
\end{proof}

Let $ g \in G_{n}. $ In order to derive an integral formula for the probability of the event $ S_{g}, $ we define the function $f_{g} : [0,1]^{n} \to [0,1]$ by
\[
f_{g}(v_{1},\ldots,v_{n}) = \prod_{ij \in g} v_{i} v_{j} \prod_{ij \in L_{n} \setminus g} [1 - (1-v_{i}) (1-v_{j})], \quad (v_{1},\ldots,v_{n}) \in [0,1]^{n}.
\]

\begin{thm}\label{thm.probstab} Consider the network model with random utilities $ (U_{n}, \mathbb{P}_{n}) $. For every $ g \in G_{n}, $ it holds that
\[
\mathbb{P}_{n}(S_{g}) = \int_{v_{1} \in [0,1]} \cdots \int_{v_{n} \in [0,1]} f_{g}(v_{1},\ldots,v_{n})dv_{n} \cdots dv_{1}.
\]
In particular, it holds that
\[
\begin{array}{rcl}
\mathbb{P}_{n}(S_{\emptyset}) & = & \int_{\bar{v}_{1} \in [0,1]} \cdots \int_{\bar{v}_{n} \in [0,1]} \prod_{ij \in L_{n}} [1 - \bar{v}_{i} \bar{v}_{j}] d\bar{v}_{n} \cdots d\bar{v}_{1}, \\
\mathbb{P}_{n}(S_{L_{n}}) & = & \tfrac{1}{n^n}.
\end{array}
\]
\end{thm}
\begin{proof}
Equation (\ref{eqn.condstab}) can be written as $\mathbb{P}_{n}(S_{g}|u(g))  = f_{g}(u(g))$. Taking the expectation with respect to $u(g)$, we obtain the first expression.

The second expression is obtained in a similar way, after a change of variables, where $ 1 - v_{i} $ is replaced by $ \bar{v}_{i} $ and $ 1 - v_{j} $ by $ \bar{v}_j, $ which are stochastic variables with the same joint distribution, as both stochastic variables are independent and uniform on $ [0,1]. $ To obtain the third expression, note that
\[
f_{L_{n}}(v_{1},\ldots,v_{n}) = \prod_{ij \in L_{n}} v_{i} v_{j} = \prod_{i=1}^{n} v_{i}^{n-1}, \quad (v_{1},\ldots,v_{n}) \in [0,1]^{n},
\]
where the second equality follows from the observation that each individual is part of exactly $ n - 1 $ links.
We have that
\[
\mathbb{P}_{n}(S_{L_{n}}) = \prod_{i=1}^{n} \int_{v_{i} \in [0,1]} v_{i}^{n-1} dv_{i} = \prod_{i=1}^{n} \tfrac{1}{n} v^{n}_{i} \vert^{1}_{0} = \prod_{i=1}^{n} \tfrac{1}{n} = \tfrac{1}{n^n}.
\]
\end{proof}

Theorem~\ref{thm.probstab} provides an exact expression for the probability that a given network $ g $ is pairwise stable. This expression is very similar to the one for the probability that a matching is stable in Equation (3) on page~95 of Knuth (1976). The main difference is the adaptation of the function $f_g$ to pairwise stability in networks, where link deletion is unilateral, while link addition requires a bilateral agreement. This expression confirms the intuition that link deletion is easier than link addition.

\begin{thm}\label{lem:monotonicity}
Consider the network model with random utilities $ (U_{n}, \mathbb{P}_{n}) $.
Let $ g \in G_{n} $ and $ ij \in L_{n} \setminus g. $ Then we have $\mathbb{P}_{n}(S_{g+ij}) \leq \mathbb{P}_{n}(S_{g}). $ In particular, the empty network $ \emptyset $ has the highest probability of being pairwise stable and the complete network $L_{n}$ the lowest.
\end{thm}
\begin{proof}
Note that, for each $v=(v_{1},\ldots,v_{n}) \in (0,1]^{n}, $

\begin{equation*}
\frac{f_{g}(v)}{f_{g+ij}(v)}  =  \frac{[1 - (1-v_{i}) (1-v_{j})]}{v_{i} v_{j}}   =  \frac{1}{v_j} + \frac{1}{v_i} - 1   \geq  1,
\end{equation*}
where the inequality uses $ 0 < v_{i}, v_{j} \leq 1$. The result now follows from Theorem~\ref{thm.probstab}.
\end{proof}

As an illustration, the next example computes the probability that particular networks are pairwise stable for the case with three individuals.
\begin{exmp}\label{exl.n=3}\rm
Let $ n = 3. $ In this case, there are $2^3=8$ different networks, the empty network, three networks with exactly one link, three networks with exactly two links, and the complete network. This results in $8!=40320$ possible strict preference orderings per individual, so to $(40320)^3$ possible strict preference profiles.

We compute the probability that a given network $ g \in G_{n} $ is pairwise stable. Using Theorem~\ref{thm.probstab}, we have
\[\mathbb{P}_{n}(S_{g}) = \int_{v_{1} \in [0,1]} \int_{v_{2} \in [0,1]} \int_{v_{3} \in [0,1]} \prod_{ij \in g} v_{i} v_{j} \prod_{ij \in L_{n} \setminus g} [1 - (1-v_{i}) (1-v_{j})]  dv_{3} dv_{2} dv_{1}.\]


Evaluating the resulting integrals, we obtain Table~\ref{table.exp}. Note that all networks with the same number of links have the same probability of being pairwise stable. This is true only in the special case $n=3$ where a network with two links always involves one player who has two links. In line with Theorem~\ref{lem:monotonicity}, the probability for pairwise stability is decreasing in the number of links.

\begin{table}[t]
\begin{center}
\renewcommand{\arraystretch}{2}
\begin{tabular}{l|cccccccc}\toprule
$g$ & $\oslash$ & $\{12\}$ & $\{13\}$ & $\{23\}$ & $\{12,13\}$ & $\{12,23\}$ & $\{13,23\}$ & $\{12,13,23\}$\\\midrule
$\mathbb{P}_{3}(S_{g})$& $\frac{50}{108}$& $\frac{19}{108}$& $\frac{19}{108}$& $\frac{19}{108}$& $\frac{8}{108}$& $\frac{8}{108}$& $\frac{8}{108}$ & $\frac{4}{108}$\\\bottomrule
\end{tabular}
\end{center}
\caption{The probability for a given network to be pairwise stable for the case $ n = 3.$}\label{table.exp}
\end{table}
The expected number of pairwise stable networks with three individuals is equal to
\[
\tfrac{50}{108}+3 \cdot \tfrac{19}{108} + 3 \cdot \tfrac{8}{108} + \tfrac{4}{108} = \tfrac{5}{4} = 1.25.
\]
\hspace*{\fill} $ \triangle $
\end{exmp}
The next example illustrates that with four individuals, two networks with the same number of links may not have the same probability of being pairwise stable.
\begin{exmp}\label{exl.n=4}\rm
Consider the networks $ g_{0} $ and $ g_{1}$:
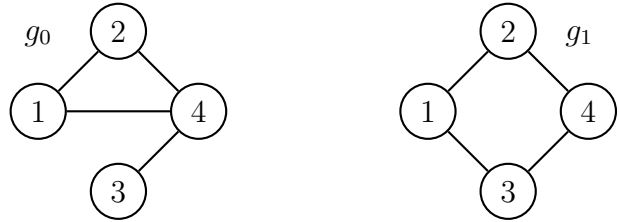
\begin{center}
\begin{tikzpicture}[node distance={15mm}, thick, main/.style = {draw, circle}]
\node[main] (1) {$1$};
\node[main] (2) [above right of=1] {$2$};
\node[main] (3) [below right of=1] {$3$};
\node[main] (4) [above right of=3] {$4$};
\draw[-] (1) -- (2);
\draw[-] (1) -- (4);
\draw[-] (2) -- (4);
\draw[-] (3) -- (4);
\node at(0,1){$g_{0}$};
\end{tikzpicture}
\hspace{2cm}
\begin{tikzpicture}[node distance={15mm}, thick, main/.style = {draw, circle}]
\node[main] (1) {$1$};
\node[main] (2) [above right of=1] {$2$};
\node[main] (3) [below right of=1] {$3$};
\node[main] (4) [above right of=3] {$4$};
\draw[-] (1) -- (2);
\draw[-] (1) -- (3);
\draw[-] (2) -- (4);
\draw[-] (3) -- (4);
\node at(2,1){$g_{1}$};
\end{tikzpicture}
\end{center}
Using Theorem~\ref{thm.probstab}, we find that
\[
\mathbb{P}_{4}(S_{g_{0}}) = \tfrac{3}{256} > \tfrac{25}{3204} = \mathbb{P}_{4}(S_{g_{1}}).
\]
\hspace*{\fill} $ \triangle $
\end{exmp}

\section{The Expected Number of Pairwise Stable Networks}\label{sec.firstmoments}
In this section, we consider the expected number of pairwise stable networks. We derive an integral formula and an exact expression for this number as well as more tractable lower and upper bounds. We establish the order of the expected number of pairwise stable networks and derive the limit of the lower and upper bounds as the number of individuals tends to infinity.
\subsection{Exact Expressions}
The number of pairwise stable networks is given by the stochastic variable
\[
s_{n} = \sum_{g \in G_{n}} 1_{S_{g}}.
\]
The next theorem presents an integral formula for the expected number of pairwise stable networks.

\begin{thm}\label{thm:expnr0}
Consider the network model with random utilities $ (U_{n}, \mathbb{P}_{n}) $. The expected number of pairwise stable networks is equal to
\[
\mathbb{E}(s_{n}) =\int_{v_{1} \in [0,1]} \cdots \int_{v_{n} \in [0,1]} \prod_{ij \in L_{n}} [v_{i}+ v_{j}] dv_{n} \cdots dv_{1}. \label{eqn.expnr0}
\]
\end{thm}
\begin{proof}
We obtain the result using the fact that the expectation of the sum of a number of stochastic variables is equal to the sum of their expectation, and compute the expectation by means of Theorem~\ref{thm.probstab}:
\begin{eqnarray*}
\mathbb{E}(s_{n}) & \hspace*{-2mm} = \hspace*{-2mm} & \mathbb{E}\big(\sum_{g \in G_{n}} 1_{S_{g}}\big) \\
 & \hspace*{-2mm} = \hspace*{-2mm} & \sum_{g \in G_{n}} \mathbb{P}_{n}(S_{g}) \\
 & \hspace*{-2mm} = \hspace*{-2mm} & \sum_{g \in G_{n}} \int_{v_{1} \in [0,1]} \cdots \int_{v_{n} \in [0,1]}  f_{g}(v_{1},\ldots,v_{n}) dv_{n} \cdots dv_{1} \\
 & \hspace*{-2mm} =  & \int_{v_{1} \in [0,1]} \cdots \int_{v_{n} \in [0,1]} \sum_{g \in G_{n}} f_{g}(v_{1},\ldots,v_{n}) dv_{n} \cdots dv_{1}.
\end{eqnarray*}
We next observe that for given $(v_{1},\ldots,v_{n}) \in \prod_{i=1}^{n} [0,1], $ it holds that
\begin{eqnarray*}
\sum_{g \in G_{n}}f_{g}(v_{1},\ldots,v_{n}) & \hspace*{-2mm} = \hspace*{-2mm} & \sum_{g \in G_{n}} \prod_{ij \in g} v_{i} v_{j} \prod_{ij \in L_{n} \setminus g} [1 - (1-v_{i}) (1-v_{j})] \\
 & \hspace*{-2mm} = \hspace*{-2mm} & \prod_{ij \in L_{n}} [v_{i} v_{j} + (1 - (1-v_{i}) (1-v_{j}))] \\
 & \hspace*{-2mm} = \hspace*{-2mm} & \prod_{ij \in L_{n}} [v_{i}+ v_{j}],
\end{eqnarray*}
where the second equality follows from the fact that the expression is a sum of terms such that for each $ ij \in L_{n} $ either $ v_{i} v_{j} $ or $1-(1-v_{i}) (1-v_{j})$ is a factor.
\end{proof}
We now use the integral in Theorem~\ref{thm:expnr0} to derive a closed-form expression for the number of pairwise stable networks. We need some notation first. Fix a network $g \in G_{n}$. We define
\[
\begin{array}{rcl}
\deg_{i}^{-}(g) & = & \#\{j \in I_{n}:i < j,~ij \in g\},\\
\deg_{i}^{+}(g) & = & \#\{j \in I_{n}:j < i,~ij \in g\},
\end{array}
\]
as the number of ``senior", respectively "junior," partners of $i$ in $g$.  The \textit{seniority degree} of player $i$ in $g$ is defined by
\[X_{i,n}(g) = 1+\deg_{i}^{-}(g)+\deg_{i}^{+}(L_{n} \setminus g).\]
In words, player $i$ scores one point for every senior partner, one point for every missing junior partner, and gets one point for free. We also define the product of the individual seniority degrees as
\[
X_{n}(g) = \prod_{i = 1}^{n}X_{i,n}(g)
\]
and refer to this number as the \it seniority degree \rm of network $ g. $

The next theorem gives an exact expression for the expected number of pairwise stable networks.
\begin{thm}\label{thm:expnr1}
Consider the network model with random utilities $ (U_{n}, \mathbb{P}_{n}) $. The expected number of pairwise stable networks is equal to
\begin{equation}\label{eqn.expnr1}
\mathbb{E}(s_{n})
= \sum_{g \in G_{n}} \frac{1}{X_{n}(g)}.
\end{equation}
\end{thm}
\begin{proof}
We first note that
\begin{equation*}
\prod_{ij \in L_{n}} [v_{i}+ v_{j}] =  \sum_{g \in G_{n}} \prod_{ij \in g} v_{i} \prod_{ij \in L_{n} \setminus g} v_{j} = \sum_{g \in G_{n}} \prod_{i = 1}^{n}v_{i}^{\deg_{i}^{-}(g)+\deg_{i}^{+}(L_{n} \setminus g)}.
\end{equation*}
We then use Theorem~\ref{thm:expnr0} and obtain
\begin{eqnarray*}
\mathbb{E}(s_{n}) & \hspace*{-2mm} = \hspace*{-2mm} & \int_{v_{1} \in [0,1]} \cdots \int_{v_{n} \in [0,1]} \sum_{g \in G_{n}} \prod_{i = 1}^{n}v_{i}^{\deg_{i}^{-}(g)+\deg_{i}^{+}(L_{n} \setminus g)}dv_{n}\cdots dv_{1} \\
 & \hspace*{-2mm} = \hspace*{-2mm} & \sum_{g \in G_{n}} \prod_{i = 1}^{n} \int_{v_{i} \in [0,1]} v_{i}^{\deg_{i}^{-}(g)+\deg_{i}^{+}(L_{n} \setminus g)}dv_{i} \\
 & \hspace*{-2mm} = \hspace*{-2mm} & \sum_{g \in G_{n}} \prod_{i = 1}^{n}\frac{1}{X_{i,n}(g)} \\
 & \hspace*{-2mm} = \hspace*{-2mm} & \sum_{g \in G_{n}} \frac{1}{X_{n}(g)}.
\end{eqnarray*}
\end{proof}
By Theorem~\ref{thm:expnr1}, the expected number of pairwise stable networks is equal to the sum over all networks of the inverse of their seniority degree. We use this expression to compute the expected number of pairwise stable networks for $ n = 2, 3, \ldots, 7. $ The results are presented in Table~\ref{T2}.

\begin{table}[t]
\begin{center}
\renewcommand{\arraystretch}{2}
\begin{tabular}{l|cccccc}
\toprule
$ n $ & 2 & 3 & 4 & 5 & 6 & 7 \\
\midrule
$ \mathbb{E}(s_{n}) $ & $ 1.0000 $ & $ 1.2500 $ & $ 2.2130 $ & $ 6.0039 $ & $ 26.3232 $ & $ 193.5913 $ \\
\bottomrule
\end{tabular}
\end{center}
\caption{The expected number of pairwise stable networks for $ n = 2,3, \ldots, 7.$}\label{T2}
\end{table}
The expected number of pairwise stable networks is reasonably small for numbers of $ n $ less than or equal to 5, but rises rapidly to $ 26.3232 $ when $ n = 6 $ and to $ 193.5913 $ when $ n = 7. $

Although the expression in Theorem~\ref{thm:expnr1} is valid for any number $ n, $ it is impractical to evaluate it for larger values of $ n $, as the number of networks is equal to $ 2^{\frac{n(n-1)}{2}}, $ which corresponds to the number of terms in the sum of (\ref{eqn.expnr1}).
We therefore continue in this section with the derivation of lower and upper bounds on $ \mathbb{E}(s_{n})$ which are easier to compute. These bounds also give us insights in the order of $ \mathbb{E}(s_{n}), $ which we show to be equal to
\[c_{n} = 2^{\ell_{n}}(\tfrac{2}{n+1})^{n} = \Big(\frac{2^{\tfrac{n+1}{2}}}{n+1}\Big)^{n}.\]
The term $2^{\ell_{n}} $ is equal to the total number of networks with $n$ individuals. Consequently, the second term governs the behaviour of the expected fraction of pairwise stable networks with $n$ individuals. 

We use $ c_{n} $ to normalize $ s_{n} $ and define the \textit{normalized number of pairwise stable networks} as the random variable
\[
s_{n}^{*} = \tfrac{1}{c_{n}} s_{n}.
\]
Table~\ref{table.n=7} displays $ \mathbb{E}(s^{*}_{n}) $ for values of $ n $ up to 7, as well as particular lower and upper bounds that we derive in the next two subsections. The last column in Table~\ref{table.n=7} corresponds to $ c_{n}. $
\begin{table}[t]
\begin{center}\renewcommand{\arraystretch}{1.2}
\begin{tabular}{ccccr}
\toprule
$ n $ & $ \mathbb{E}(s_{n}^{*}) $ & $ e^{\tfrac{1}{2} \tfrac{(n\!-\!1)n}{(n\!+\!1)^{2}}} $ & $\big(1+\tfrac{1}{n}\big)^{n}\big(1-\tfrac{1}{2^{n}}\big)^{n}$ & $ c_{n} $ \\
\midrule
2 & 1.1250 & 1.1176 & 1.2656 & 0.8889 \\
3 & 1.2500 & 1.2062 & 1.5880 & 1.0000 \\
4 & 1.3507 & 1.2712 & 1.8859 & 1.6384 \\
5 & 1.4248 & 1.3202 & 2.1231 & 4.2140 \\
6 & 1.4767 & 1.3582 & 2.2943 & 17.8255\\
7 & 1.5124 & 1.3884 & 2.4112 & 128.0000 \\
$\cdots$ & $\cdots$ & $\cdots$ & $\cdots$ & $ \cdots $ \\
$\infty$ & & 1.6487 & 2.7182 & \\
\bottomrule
\end{tabular}
\end{center}
\caption{The expected normalized number of pairwise stable networks, a lower bound, an upper bound, and the normalizing constant.}\label{table.n=7}
\end{table}
\subsection{A Lower Bound on $ \mathbb{E}(s^{*}_{n}) $}
We first turn our attention to the derivation of a lower bound for the expected normalized number of pairwise stable networks $ \mathbb{E}(s^{*}_{n}). $

Towards deriving the lower bound, we interpret $ X_{i,n} $ and $ X_{n} $ as stochastic variables on $G_{n}$, where $ G_{n} $ is endowed with a uniform probability measure. Equivalently, we can think of the Erdős–Rényi–Gilbert model of random graphs on the set of vertices $ I_{n}, $ where each link $ ij $ is chosen with probability $ 1/2, $ independently of the other links. In this model, all graphs in $G_{n}$ are equally likely. We let $\mathbb{E}_{\mathcal{G}}$ denote the corresponding expectation operator.

We then rewrite Equation~\eqref{eqn.expnr1} as
\[
\mathbb{E}(s_{n}) = 2^{\ell_{n}} \mathbb{E}_{\mathcal{G}}(\tfrac{1}{X_{n}}),
\]
or, equivalently,
\begin{equation}
\mathbb{E}(s_{n}^{*}) = (\tfrac{n+1}{2})^{n} \mathbb{E}_{\mathcal{G}}(\tfrac{1}{X_{n}}). \label{E4.2}
\end{equation}

\begin{thm}\label{thm:numberlower}
Consider the network model with random utilities $ (U_{n}, \mathbb{P}_{n}) $. A lower bound on the normalized expected number of pairwise stable networks $ \mathbb{E}(s_{n}^{*}) $ is given by
\[
e^{\tfrac{1}{2} \tfrac{(n-1)n}{(n+1)^{2}}}.
\]
\end{thm}
\begin{proof}
We prove the lemma by showing that
\[\ln \mathbb{E}(s_{n}^{*}) \geq \tfrac{1}{2} \tfrac{(n-1)n}{(n+1)^{2}}.\]
It holds that
\begin{equation}
\ln (\mathbb{E}_{\cal G}(\tfrac{1}{X_{n}}))  \geq \mathbb{E}_{\cal G}(\ln[\tfrac{1}{{X}_{n}}]) = -\sum_{i=1}^{n} \mathbb{E}_{\cal G}(\ln[X_{i,n}]) = - n \mathbb{E}_{\cal G}(\ln[X_{1,n}]), \label{E4.3}
\end{equation}
where the inequality follows from Jensen's inequality, and the final equality from the fact that the stochastic variables $X_{1,n},\ldots,X_{n,n}$ have the same distribution and hence have the same mean.

The stochastic variable $ X_{1,n} - 1 $ has the same distribution as the degree of individual~1 at a graph in $ G_{n} $ and is therefore equal to $ \mathsf{Binom}(n-1,1/2). $

Let $ B $ be a stochastic variable with a $\mathsf{Binom}(n-1,1/2)$ distribution. We show next that
\begin{equation}
\mathbb{E}(\ln(1+B)) \leq \ln(\tfrac{n+1}{2}) - \tfrac{1}{2} \tfrac{n-1}{(n+1)^{2}}. \label{E4.4}
\end{equation}
It holds that $ \mathbb{E}(B) = (n-1)/2, $ which we denote by $ \mu. $
For $ x > -1, $ the third-order Taylor expansion of $ \ln(1+x) $ around $ \mu $ is equal to
\[
\ln(1+\mu) + \tfrac{2}{n+1} (x-\mu) + \tfrac{1}{2} \tfrac{-4}{(n+1)^{2}} (x-\mu)^{2} + \tfrac{1}{6} \tfrac{16}{(n+1)^{3}} (x-\mu)^{3}.
\]
Since the fourth-order term in the Taylor expansion is less than or equal to zero, it holds that $ \ln(1+x) $ is dominated by the third-order Taylor expansion.
We find that
\[
\begin{array}{rcl}
\mathbb{E} (\ln(1+B)) & \leq & \mathbb{E} \Big(\ln(1+\mu) + \tfrac{2}{n+1} (B-\mu) + \tfrac{1}{2}  \tfrac{-4}{(n+1)^{2}} (B-\mu)^{2} + \tfrac{1}{6} \tfrac{16}{(n+1)^{3}} (B-\mu)^{3}\Big) \\
 & = & \ln(\tfrac{n+1}{2}) + \tfrac{1}{2} \tfrac{-4}{(n+1)^{2}} \tfrac{n-1}{4} \\
 & = & \ln(\tfrac{n+1}{2}) - \tfrac{1}{2} \tfrac{n-1}{(n+1)^{2}},
\end{array}
\]
where the first equality uses that $ \mathbb{E}(B) = (n-1)/2, $ $ \mathbb{E}(B-\mu) = 0, $ $ \mathbb{E}(B-\mu)^{2} = (n-1)/4, $ and $ \mathbb{E}(B-\mu)^{3} = 0. $

We now obtain that
\[
\begin{array}{rcl}
\ln\mathbb{E}(s_{n}^{*}) & = & \ln(\tfrac{n+1}{2})^{n} + \ln (\mathbb{E}_{\cal G}(\tfrac{1}{X_{n}})) \\
 & \geq & \ln(\tfrac{n+1}{2})^{n} - n \mathbb{E}_{\cal G}(\ln[X_{1,n}]) \\
 & \geq & \ln(\tfrac{n+1}{2})^{n} - \ln(\tfrac{n+1}{2})^{n} + \tfrac{1}{2} \tfrac{(n-1)n}{(n+1)^{2}} \\
 & = & \tfrac{1}{2} \tfrac{(n-1)n}{(n+1)^{2}},
\end{array}
\]
where the first equality follows from (\ref{E4.2}), the first inequality from (\ref{E4.3}), and the second inequality from (\ref{E4.4}).
\end{proof}
The lower bound on $ \mathbb{E}(s^{*}_{n}) $ in Theorem~\ref{thm:numberlower} is presented in the third column of Table~\ref{table.n=7}. It is less than or equal to $ \sqrt{e} $ for all values of $ n $ and converges to $ \sqrt{e} $ as $ n $ tends to infinity. 
\subsection{An Upper Bound on $ \mathbb{E}(s^{*}_{n}) $}
We now derive an upper bound on the expected normalized number of pairwise stable networks $ \mathbb{E}(s^{*}_{n}). $
\begin{thm}\label{thm:numberupper}
Consider the network model with random utilities $ (U_{n}, \mathbb{P}_{n}) $. An upper bound on the normalized expected number of pairwise stable networks $ \mathbb{E}(s_{n}^{*}) $ is given by
\[
\big(1+\tfrac{1}{n}\big)^{n} (1-\tfrac{1}{2^{n}})^{n}.
\]
\end{thm}
\begin{proof}
For every $ij \in L_{n}$ it holds that $v_{i} \leq 1$ and $v_{j} \leq 1$, so $v_{i} + v_{j} \leq 1 + v_{i} v_{j}, $ which is equivalent to
\[
v_{i} + v_{j} \leq\big(\tfrac{1+v_{i}}{\sqrt{2}}\big) \big(\tfrac{1+v_{j}}{\sqrt{2}}\big).
\]
Using Theorem~\ref{thm:expnr0}, we obtain
\begin{eqnarray*}
\mathbb{E}(s_{n}) & \hspace*{-2mm} = \hspace*{-2mm} & \int_{v_{1} \in [0,1]} \cdots \int_{v_{n} \in [0,1]} \prod_{ij \in L_{n}} [v_{i}+ v_{j}] dv_{n} \cdots dv_{1} \\
 & \hspace*{-2mm} \leq \hspace*{-2mm} & \int_{v_{1} \in [0,1]} \cdots \int_{v_{n} \in [0,1]} \prod_{ij \in L_{n}} \big(\tfrac{1+v_{i}}{\sqrt{2}}\big) \big(\tfrac{1+v_{j}}{\sqrt{2}}\big) dv_{n} \cdots dv_{1} \\
 & \hspace*{-2mm} = \hspace*{-2mm} & \int_{v_{1} \in [0,1]} \cdots \int_{v_{n} \in [0,1]} \prod_{i=1}^{n} \big(\tfrac{1+v_{i}}{\sqrt{2}}\big)^{n-1} dv_{n} \cdots dv_{1} \\
 & \hspace*{-2mm} = \hspace*{-2mm} & \prod_{i=1}^{n} \int_{v_{i} \in [0,1]} \big(\tfrac{1+v_{i}}{\sqrt{2}}\big)^{n-1} dv_{i} \\
 & \hspace*{-2mm} = \hspace*{-2mm} & \Big(\tfrac{1}{n} 2^{\tfrac{1}{2}(n+1)} (1 - \tfrac{1}{2^{n}})\Big)^{n} \\
 & \hspace*{-2mm} = \hspace*{-2mm} & \tfrac{1}{n^{n}} 2^{\ell_{n+1}}(1-\tfrac{1}{2^{n}})^{n}.
\end{eqnarray*}
After a division by $ c^{n} $ to normalize, we find
\[\mathbb{E}(s_{n}^{*}) = \tfrac{1}{c_{n}} \mathbb{E}(s_{n}) \leq (1 + \tfrac{1}{n})^{n}(1-\tfrac{1}{2^{n}})^{n}.\]
\end{proof}
The upper bound on $ \mathbb{E}(s^{*}_{n}) $ in Theorem~\ref{thm:numberupper} is presented in the fourth column of Table~\ref{table.n=7}. It is less than or equal to $ e $ for all values of $ n $ and converges to $ e $ as $ n $ tends to infinity.
\subsection{Asymptotic Behavior of $ \mathbb{E}(s^{*}_{n}) $}
As we already remarked, by taking limits as $ n $ goes to infinity in the expressions of Theorems~\ref{thm:numberlower} and \ref{thm:numberupper}, the limit of the expectation of $ s_{n}^{*} $ lies between $ \sqrt{e} $ and $ e. $ We thus obtain Theorem~\ref{T4.1}.
\begin{thm} \label{T4.1}
For every $ n \in \mathbb{N}, $ consider the network model with random utilities $ (U_{n}, \mathbb{P}_{n}) $.
It holds that
\[
\sqrt{e} \leq \liminf_{n \to \infty} \mathbb{E}(s_{n}^{*}) \leq \limsup_{n \to \infty} \mathbb{E}(s_{n}^{*}) \leq e.
\]
\end{thm}
The expected number of pairwise stable networks goes to infinity if $n$ goes to infinity, with speed comparable to that of the sequence $c_{n}$. On the other hand, the fraction of pairwise stable networks within the entire set of networks is of the order $ (\tfrac{2}{n+1})^{n} $ and thus approaches zero.
\section{The Variance of the Number of Pairwise Stable Networks}\label{sec.secondmoments}

In this section, we first derive the relationship between the Hamming distance between two networks $ g_{0} $ and $ g_{1} $ and the conditional dependence of the events $ S_{g_{0}} $ and $ S_{g_{1}}. $ We use this relationship to estimate the variance of the number of pairwise stable networks and show that it converges to zero if $ n $ tends to infinity.

\subsection{Hamming Distance and Dependencies}

The Hamming distance $ d $ on the set of networks $ G_{n} $ is given by
\[
d(g_{0},g_{1}) = \#((g_{0} \setminus g_{1}) \cup (g_{1} \setminus g_{0})), \quad g_{0}, g_{1} \in G_{n}.
\]
For $ r \in \mathbb{N}, $ let $B_{r}(g_{0}) = \{g_{1} \in G_{n}: d(g_{0},g_{1}) \leq r\}$ denote the ball of radius $ r $ around $ g_{0} \in G_{n}. $ A set $ N \subseteq G_{n}$ is called an $r$-net if for any two distinct $ g_{0}, g_{1} \in N $ it holds that $ d(g_{0},g_{1}) \geq r. $ These notions are helpful since the pairwise stability of a network $ g_{0} $ only involves networks which differ in a single link from $ g_{0}, $ i.e., networks in a $1$-ball around $g_0$. We can therefore describe the event $ S_{g_{0}} $ as in the next lemma.

\begin{lem} \label{L4.12}
Consider the network model with random utilities $ (U_{n}, \mathbb{P}_{n}) $.
The event that $ g_{0} \in G_{n} $ is pairwise stable is equal to
\[
S_{g_{0}} = \cap_{g_{1} \in B_{1}(g_{0})} D^{\rm c}(g_{0},g_{1}).
\]
\end{lem}
It follows from Lemma~\ref{L4.12} that $ S_{g_{0}} $ is measurable with respect to the sigma-algebra generated by the stochastic variables $\{u(g_{1}): g_{1} \in B_{1}(g_{0})\}$.

We now delve deeper into the conditional dependence of the events $ S_{g_{0}} $ and $ S_{g_{1}}. $ In general, the event that $g_{0}$ is pairwise stable is not independent of the event that $g_{1}$ is pairwise stable. Lemma~\ref{L4.13} links the nature of this dependence to the Hamming distance between the two networks.
\begin{lem}\label{L4.13} Consider the network model with random utilities $ (U_{n}, \mathbb{P}_{n}) $. Let $ g_{0}, g_{1} \in G_{n}. $
\begin{itemize}
\item[(i)] If $ d(g_{0},g_{1}) = 1, $ then the event $ S_{g_{0}} \cap S_{g_{1}} $ has measure zero. Consequently, the set of pairwise stable networks is a 2-net almost surely.
\item[(ii)] If $d(g_{0},g_{1}) = 2$, then the events $S_{g_{0}}$ and $S_{g_{1}}$ are positively dependent. More generally, if $ N \subseteq G_{n} $ is a 2-net, the events in $ \{S_{g} : g \in N\} $ are positively dependent in the following sense: there exist stochastic variables $\nu_{1},\ldots,\nu_{k}$ on $U$ such that
\begin{itemize}
\item[(iia)] the events in $ \{S_{g}: g \in N\} $  are independent conditional on $\nu_{1},\ldots,\nu_{k}, $
\item[(iib)] for each $ g \in N, $ the conditional probability $\mathbb{P}_{n}(S_{g}|\nu_{1},\ldots,\nu_{k})$ is a non-decreasing function of each of the stochastic variables $\nu_{1},\ldots,\nu_{k}$.
\end{itemize}
\item[(iii)] If $ d(g_{0},g_{1}) \geq 3, $ then the events $S_{g_{0}}$ and $S_{g_{1}}$ are independent. More generally, if $ N \subseteq G_{n} $ is a 3-net, then the events $ \{S_{g}: g \in N\} $ are independent.
\end{itemize}
\end{lem}
\begin{proof}
(i) If $ d(g_{0},g_{1}) = 1, $ then the two events $ D(g_{1},g_{0}) $ and $ D^{\rm c}(g_{0},g_{1}) $ differ only when there are ties. Ties occur with probability $0$.\medskip

\noindent(ii) Let $ \nu_{1}, \ldots, \nu_{k} $ be the stochastic variables in the collection $ \{-u_{i}(g) : g \in G_{n} \setminus N, $ $ i \in I_{n}\}. $ Then condition (iia) follows from the fact that $ N $ is a 2-net and from Lemma~\ref{L4.12}. Condition (iib) holds since the indicator function $ 1_{S_{g}} : U \to [0,1]$ is non-decreasing in $ -u_{i}(g^{\prime})$, for each $ i \in I_{n} $ and each network $ g^{\prime} $ other than $ g. $ So in particular, it is non-decreasing in $\nu_{1},\ldots,\nu_{k}$.
\medskip

\noindent(iii) This follows from Lemma~\ref{L4.12} and the fact that $ B_{1}(g_{0}) $ and $ B_{1}(g_{1}) $ are disjoint.
\end{proof}
If the Hamming distance between $g_{0}$ and $g_{1}$ is exactly 1, then the events $S_{g_{0}}$ and $S_{g_{1}}$ are mutually exclusive: If $ g_{0} $ is pairwise stable, then it is not dominated by $g_{1}$; but this means, modulo indifferences, that $g_{0}$ dominates $g_{1}$, and so $g_{1}$ is not pairwise stable.

If the Hamming distance between $g_{0}$ and $g_{1}$ is 3 or more, then the events $S_{g_{0}}$ and $S_{g_{1}}$ are independent. To verify the pairwise stability of $ g_{0} $ does not involve any network that is needed for the verification of the pairwise stability of $g_{1}$.

When $g_{0}$ and $g_{1}$ are exactly at a Hamming distance of 2 from each other, the events $S_{g_{0}}$ and $S_{g_{1}}$ are positively correlated. Intuitively, knowing that $ g_{0} $ is pairwise stable reveals that other networks in its 1-ball are relatively ``poor". Two such networks are also in the 1-ball around $ g_{1}. $ This information, therefore, enhances the likelihood that $g_{1}$ dominates these two networks.

The next example builds additional intuition for the result in Lemma~\ref{L4.13} and illustrates this lemma for the case $ n = 3. $

\begin{exmp}\rm
Let $ n = 3. $ For any two distinct networks $ g_{0}, g_{1} \in G_{3}, $ we compute the ratios of probabilities
\[
\displaystyle\frac{\mathbb{P}_{3}(S_{g_{0}} \cap S_{g_{1}})}{\mathbb{P}_{3}(S_{g_{0}}) \mathbb{P}_{3}(S_{g_{1}})}.
\]
 The results are reported in Table~\ref{table.correlation}.
\begin{table}[t]
\begin{center}
\renewcommand{\arraystretch}{1.5}
\begin{tabular}{lccccccc}
\toprule
 & $\{12\}$ & $\{13\}$ & $\{23\}$ & $\{12,13\}$ & $\{12,23\}$ & $\{13,23\}$ & $\{12,13,23\}$ \\
\midrule

$\varnothing$
    & 0 & 0 & 0
    & 
    1.2488
    & 
    1.2488
    & 
    1.2488
    & 
    1 \\

$\{12\}$
    &
    & 
    1.3172
    & 
    1.3172
    & 0
    & 0
    & 
    1
    & 
    1.1053 \\

$\{13\}$
    &
    &
    & 
    1.3172
    & 0
    & 
    1
    & 0
    & 
    1.1053 \\

$\{23\}$
    &
    &
    &
    & 
    1
    & 0
    & 0
    & 
    1.1053 \\

$\{12,13\}$
    &
    &
    &
    &
    & 
    1.3125
    & 
    1.3125
    & 0 \\

$\{12,23\}$
    &
    &
    &
    &
    &
    & 
    1.3125
    & 0 \\

$\{13,23\}$
    &
    &
    &
    &
    &
    &
    & 0 \\

\bottomrule
\end{tabular}
\end{center}
\caption{The ratios $ \displaystyle\frac{\mathbb{P}_{3}(S_{g_{0}} \cap S_{g_{1}})}{\mathbb{P}_{3}(S_{g_{0}}) \mathbb{P}_{3}(S_{g_{1}})} $ for any two distinct networks $ g_{0}, g_{1} \in G_{3}. $}\label{table.correlation}
\end{table}
The ratio is 0 for networks which are at Hamming distance 1 from each other by Lemma~\ref{L4.13}.(i). The 1's in the matrix are implied by statistical independence which occurs at Hamming distance 3 by Lemma~\ref{L4.13}.(iii). Networks of Hamming distance~2 have a positive correlation by Lemma~\ref{L4.13}.(ii). Since we have computed the denominators of the expression in Table~\ref{table.correlation} in Section 3, it remains to compute the numerators, where we illustrate one sample computation of the joint event that two networks at Hamming distance~2 are pairwise stable.

Consider the networks $g_{0} = \{12\}$ and $g_{1} = \{12,13,23\}$. Following an approach similar to that in Section~\ref{sec.structure}, we first calculate the conditional probability for $g_{0}$ and $g_{1}$ to be both pairwise stable, given their utilities $ u(g_{0}), u(g_{1}) \in \mathbb{R}^{I_{n}}. $ To simplify notation, we write $ v^{0}_i $ for $ u_{i}(g_{0}) $ and $ v^{1}_i $ for $ u_{i}(g_{1}). $

The event that both $g_{0}$ and $g_{1}$ are simultaneously pairwise stable occurs if there is no profitable deviation to a network within a 1-ball around either $ \{12\} $ or $ \{12,13,23\}. $ Each of these two 1-balls has cardinality $4$, including the network itself, leaving us with $2 \cdot 3 = 6 $ possible deviations. They have two networks in common, namely the networks $\{12, 13\}$ and $\{12, 23\}$. Taking this into account, we consider the four events listed in Table~\ref{T4}.
\begin{table}[t]
\begin{center}
\renewcommand{\arraystretch}{1.25}
\begin{tabular}{llc}
\toprule
 & No deviation & Conditional probability given $ v^{0} $ and $ v^{1} $\\
\midrule
1 & from $\{12\}$ to $\emptyset$ & $ v^{0}_{1} v^{0}_{2} $\\
2 & from $\{12\}$ or $\{12, 13, 23\}$ to $\{12, 13\}$ & $ v^{0}_{1} v^{1}_{2} v^{1}_{3} + (1-v^{0}_{1}) v^{1}_{2} \min\{v^{0}_{3}, v^{1}_{3}\} $\\
3 & from $\{12\}$ or $\{12, 13, 23\}$ to $\{12, 23\}$ & $ v^{0}_{2} v^{1}_{1} v^{1}_{3} + (1-v^{0}_{2}) v^{1}_{1} \min\{v^{0}_{3}, v^{1}_{3}\} $\\
4 & from $\{12, 13, 23\}$ to $\{13,23\}$ & $ v^{1}_{1} v^{1}_{2}$\\
\bottomrule
\end{tabular}
\end{center}
\caption{Conditional probabilities of absence of particular deviations.} \label{T4}
\end{table}
The four events are conditionally independent given $ v^{0} $ and $ v^{1}. $ Consequently, to compute the probability of the event $S_{g_{0}} \cap S_{g_{1}}$ we multiply the four entries of the table and take the integral of their product with respect to $ v^{0}_{1}, v^{0}_{2}, v^{0}_{3}, v^{1}_{1}, v^{1}_{2}, $ and $ v^{1}_{3}. $

The other computations for networks at Hamming distance~2 are similar, leading to the results reported in Table~\ref{table.correlation}.
\hspace*{\fill} $ \triangle $
\end{exmp}
\subsection{Convergence of the Variance}
In case $ n=3, $ there are few pairs of networks for which the probability that both are pairwise stable is independent. Yet, if $ n $ tends to infinity, the fraction of such pairs of networks goes to one. This vanishing dependence enables us to show that in the limit the variance of the normalized number of pairwise stable networks goes to zero.


\begin{thm}\label{thm.variance}
Consider the network model with random utilities $ (U_{n}, \mathbb{P}_{n}) $. The variance $ \sigma^{2}(s_{n}) $ of the number and $ \sigma^{2}(s^{*}_{n}) $ of the normalized number of pairwise stable networks satisfies
\[
\begin{array}{l}
\sigma^{2}(s_{n}) \leq \ell_{n}^{2} \mathbb{E}(s_{n}),\\
\sigma^{2}(s_{n}^{*}) \leq \tfrac{\ell_{n}^{2}}{c_{n}} \mathbb{E}(s_{n}^{*}).
\end{array}
\]
In particular, it holds that
\[
\lim_{n \to \infty} \sigma^{2}(s_{n}^{*}) = 0.
\]
\end{thm}
\begin{proof}
The second moment of the number of pairwise stable networks is equal to
\begin{eqnarray*}
\mathbb{E}(s_{n}^{2}) & \hspace*{-2mm} = \hspace*{-2mm} & \mathbb{E}\left(\sum_{g_{0},g_{1} \in G_{n}}1_{S_{g_{0}} \cap S_{g_{1}}}\right)\\
 & \hspace*{-2mm} = \hspace*{-2mm} & \sum_{\substack{g_{0},g_{1} \in G_{n}\\ d(g_{0},g_{1}) \leq 2}} \mathbb{P}_{n}(S_{g_{0}} \cap S_{g_{1}}) + \sum_{\substack{g_{0},g_{1} \in G_{n}\\ d(g_{0},g_{1}) \geq 3}} \mathbb{P}_{n}(S_{g_{0}}) \mathbb{P}_{n}(S_{g_{1}})\\
 & \hspace*{-2mm} \leq \hspace*{-2mm} & \sum_{\substack{g_{0},g_{1} \in G_{n}\\ d(g_{0},g_{1}) \leq 2}} \mathbb{P}_{n}(S_{g_{0}}) + (\mathbb{E}(s_{n}))^{2}\\
 & \hspace*{-2mm} \leq \hspace*{-2mm} & \ell_{n}^{2} \sum_{g_{0} \in G_{n}} \mathbb{P}_{n}(S_{g_{0}}) + (\mathbb{E}(s_{n}))^{2} \\
 & \hspace*{-2mm} = \hspace*{-2mm} & \ell_{n}^{2} \mathbb{E}(s_{n}) + (\mathbb{E}(s_{n}))^{2},
\end{eqnarray*}
where the first equality follows directly from the definition of $ s_{n}^{2}, $ the second equality uses Lemma~\ref{L4.13}.(iii), the inequality in the third line uses $ \mathbb{P}_{n}(S_{g_{0}} \cap S_{g_{1}}) \leq \mathbb{P}_{n}(S_{g_{0}}) $ and
\[
\sum_{\substack{g_{0},g_{1} \in G_{n}\\ d(g_{0},g_{1}) \geq 3}} \mathbb{P}_{n}(S_{g_{0}}) \mathbb{P}_{n}(S_{g_{1}}) \leq \sum_{g_{0},g_{1} \in G_{n}} \mathbb{P}_{n}(S_{g_{0}}) \mathbb{P}_{n}(S_{g_{1}}) = (\mathbb{E}(s_{n}))^{2},
\]
and the inequality in the fourth line follows since the cardinality of the $2$-ball around any $g_{0} \in G_{n}$ is equal to $ 1 + \ell_{n} + \ell_{n}(\ell_{n}-1)/2, $ so at most $\ell_{n}^{2}. $
We find that
\[
\sigma^{2}(s_{n}) = \mathbb{E}(s_{n}^{2}) - (\mathbb{E}(s_{n}))^{2} \leq \ell_{n}^{2} \mathbb{E}(s_{n}).
\]
Next, we have that
\[
\begin{array}{rcl}
\sigma^{2}(s_{n}^{*}) & = & \sigma^{2}(\frac{1}{c_{n}} s_{n}) = \frac{1}{c^{2}_{n}} \sigma^{2}(s_{n}) \leq \frac{\ell^{2}_{n}}{c^{2}_{n}} \mathbb{E}(s_{n}) = \frac{\ell^{2}_{n}}{c_{n}} \mathbb{E}(s_{n}^{*}).
\end{array}
\]
Finally, it holds that
\[
\begin{array}{rcl}
\lim_{n \rightarrow \infty} \frac{\ell^{2}_{n}}{c_{n}} & = & \lim_{n \rightarrow \infty} (\frac{n(n-1)}{2})^{2} \left(\frac{n+1}{2^{\frac{n+1}{2}}}\right)^{n} = 0.
\end{array}
\]
By Theorem~\ref{T4.1}, we have $ \limsup_{n \to \infty} \mathbb{E}(s_{n}^{*}) \leq e. $ We conclude that $ \lim_{n \to \infty} \sigma^{2}(s_{n}^{*}) = 0. $
\end{proof}

\subsection{Two Implications for Almost-sure Convergence}
We conclude our analysis of the second moments with two straightforward consequences. The result below amounts to saying that, with probability one, the number of pairwise stable networks converges to infinity while the share of pairwise stable networks converges to zero, as $n$ goes to infinity. These results follow directly from Theorem \ref{thm.variance} by applying Chebyshev's inequality.

A technical detail is in order. We have introduced the random variable $s_{n}$ defined, for a given $n \in \N$, on the probability space $(U_{n},\pr_{n})$. To speak of ``convergence'' of random variabes, we must couple them, i.e., define the copies of these random variables on a common probability space. 

To do so, interpret $\N= \{1,2,\ldots\}$ as a universal population of individuals.  Let $L$ be the set of links between the individuals in $\N$. A network for the universal population is any finite subset of $L$. Let $G$ denote the set of networks in the universal population. A set of links $L_{n}$ in the population $I_{n}$ is then a subset of $L$, and, likewise, a set of networks $G_{n}$ in the population $n$ is a subset of $G$.  A universal network model is a function $u : \N \times G \to [0,1]$. Note that a network model of size $n$ is the restriction of $u$ to the subset of $G_{n} \times I_{n}$ of $G \times \N$. Let $U = [0,1]^{\N\times G}$ be the set of universal network models.

Noting that $G$ is a countable set, we endow $U$ with the product $\pr$ of uniform distributions on $[0,1]$ and refer to the probability space $(U,\pr)$ as the \textit{universal network model with random utilities}. The probability space  $(U,\pr)$ carries the copies of the random variables $s_{n}$ for all $n \in \N$ as well as all the other random variables we have considered. We continue to use the same symbol to denote the copies.



\begin{cor}\label{cor.as}
Consider the universal network model with random utilities  $(U,\pr)$. With probability 1, the number of pairwise stable networks $s_{n}$ converges to infinity, while the share of pairwise stable networks $2^{-\ell_{n}}s_{n}$ converges to zero as $n$ goes to infinity.
\end{cor}

The result is stated for the universal network model merely for concreteness. Theorem \ref{thm.variance} remains valid for any coupling of $s_{n}$.

\section{Discussion}
This paper provides a first step in analyzing random networks. Within our framework, an open question remains whether the normalized expected number of pairwise stable networks converges, and if so, to compute the limit. By our results on the variance, for convergence in $L_2$ to a constant, it suffices to show monotonicity of the sequence $ \mathbb{E}(s_{n}^{*}) $.

There are many other directions for future research. Following the approach of Rinott and Scarsini (2000), one could introduce dependence in the payoffs of the players at a given network. Other cases with such dependencies would arise if one considers typical network models like the co-author model or the connections model of Jackson and Wolinsky (1996) and assume that each node in the network has a randomly drawn value. Numerical work on the connections model with heterogeneous valuations in Herings and Zhan (2024) provides some evidence that there are many pairwise stable networks in this case, too.

One could also consider alternative solution concepts in networks. One possibility would be to consider the case of weighted pairwise stable networks as in Bich and Morhaim (2020). Since they cover pairwise stable networks as a special case, the expected number of weighted pairwise stable networks exceeds that of pairwise stable networks. Existence of weighted pairwise stable networks is not an issue. 

Under pairwise stability, a player considers the deletion of a single link and a pair of players the addition of a single link. Strongly stable networks as defined in Jackson and van den Nouweland (2005) is a much stronger concept that allows arbitrary coalitions to deviate without restrictions on the number of links that are modified. The probability that a strongly stable network exists and the expected number of such networks are therefore natural questions to consider for the network model with random utilities.

\section*{References}
\begin{description}
\item[\sc Amiet, B., Collevecchio, A., Scarsini, M., and Z. Zhong (2021),] ``Pure Nash Equilibria and Best-response Dynamics in Random Games," \it Mathematics of Operations Research, 46, \rm 1552--1572.
\item[\sc Ashlagi, I., Y. Kanoria, and J.D. Leshno (2017),] ``Unbalanced Random Matching Markets: The Stark Effect of Competition," \it Journal of Political Economy, 125, \rm 69--98.
\item[\sc B\a'{a}r\a'{a}ny, I., S. Vempala, and A. Vetta (2007),] ``Nash Equilibria in Random Games," \it Random Structures and Algorithms, 31, \rm 391--405.
\item[\sc Bich, P., and L. Morhaim (2020),] ``On the Existence of Pairwise Stable Weighted Networks," \it Mathematics of Operations Research, 45, \rm 1393--1404.
\item[\sc Dresher, M. (1970),] ``Probability of a Pure Equilibrium Point in $n$-Person Games," \it Journal of Combinatorial Theory, 8, \rm 134--145.
\item[\sc Gale, D., and L.S. Shapley (1962),] ``College Admissions and the Stability of Marriage," \it American Mathematical Monthly, 69, \rm 9--15.
\item[\sc Goldberg, K., A.J. Goldman, and M. Newman (1968),] ``The Probability of an Equilibrium Point," \it Journal of Research of the National Bureau of Standards - B. Mathematical Sciences, 72B, \rm 93-–101.
\item[\sc Goldman, A.J. (1957),] ``The Probability of a Saddlepoint," \it American Mathematical Monthly, 64, \rm 729--730.
\item[\sc Hellmann, T. (2013),] ``On the Existence and Uniqueness of Pairwise Stable Networks," \it International Journal of Game Theory, 42, \rm 211--237.
\item[\sc Herings, P.J.J., and Y. Zhan (2024),] ``The Computation of Pairwise Stable Networks," \it Mathematical Programming, 203, \rm 443--473.
\item[\sc Jackson, M.O., and A. van den Nouweland (2005),] ``Strongly Stable Networks," \it Games and Economic Behavior, 51, \rm 420--444.
\item[\sc Jackson, M.O., and A. Watts (2001),] ``The Existence of Pairwise Stable Networks," \it Seoul Journal of Economics, 14, \rm 299--321.
\item[\sc Jackson, M.O., and A. Wolinsky (1996),] ``A Strategic Model of Social and Economic Networks," \it Journal of Economic Theory, 71, \rm 44--74.
\item[\sc Knuth, D.E. (1976),] \it Marriages Stables, \rm Les Presses de l'Universit\a'{e} de Montreal, Montreal.
\item[\sc McLennan, A. (2005),] ``The Expected Number of Nash Equilibria of a Normal Form Game," \it Econometrica, 73, \rm 141--174.
\item[\sc McLennan, A., and J. Berg (2005),] ``Asymptotic Expected Number of Nash Equilibria of Two-player Normal Form Games," \it Games and Economic Behavior, 51, \rm 264--295.
\item[\sc Pittel, B. (1989),] ``The Average Number of Stable Matchings," \it SIAM Journal on Discrete Mathematics, 2, \rm 530--549.
\item[\sc Rinott, Y., and M. Scarsini (2000),] ``On the Number of Pure Strategy Nash Equilibria in Random Games," \it Games and Economic Behavior, 33, \rm 274--293.
\end{description}

\end{document}